\documentclass[UKenglish,copyright,creativecommons]{eptcs}
\providecommand{\event}{PLACES 2016} 
\usepackage{doc}
\usepackage{makeidx}
\usepackage{color}
\usepackage{stmaryrd}
\usepackage{xspace}
\usepackage{mathpartir,amsmath,amssymb,amsthm}
\usepackage{prooftree}
\usepackage{mathtools}
\usepackage{macros}

\title{Secure Multiparty Sessions with Topics\thanks{Partly supported
    by the COST Action IC1201 BETTY.}
}
\author{Ilaria Castellani 
\institute{INRIA Sophia Antipolis, France}
\and
Mariangiola Dezani-Ciancaglini\thanks{Partly supported by EU H2020-644235 Rephrase project, EU H2020-644298 HyVar project, ICT COST Actions IC1402 ARVI and Ateneo/CSP project RunVar.} 
\institute{University of Turin, Italy }
\and
Ugo de'Liguoro
\thanks{Partly supported by EU H2020-644235 Rephrase project, EU H2020-644298 HyVar project, ICT COST Actions IC1402 ARVI and Ateneo/CSP project RunVar.
}
\institute{University of Turin, Italy}
}
\def\titlerunning{Secure multiparty sessions with topics}
\def\authorrunning{I. Castellani, M. Dezani-Ciancaglini \& U. de'Liguoro}
\def\copyrightholders{Castellani, Dezani-Ciancaglini \& de'Liguoro}
\providecommand{\publicationstatus}{To appear in EPTCS.}
\begin{document}
\clearpage\maketitle
\thispagestyle{empty}

\begin{abstract}
  Multiparty session calculi have been recently equipped with security
  requirements, in order to guarantee 
  properties such as access control
  and 
  leak freedom. However, the proposed security
  requirements seem to be overly restrictive in some cases. In
  particular, a party is not allowed to communicate any kind of public
  information after receiving a secret information. This does not seem
  justified in case the two pieces of information are totally
  unrelated. The aim of the present paper is to overcome this
  restriction, by designing a type discipline for a simple multiparty
  session calculus, which classifies messages according to their
  topics and allows unrestricted sequencing of messages on independent
  topics.
\end{abstract}









\mysection{Introduction}

Today's distributed computing environment strongly relies on
communication. Communication often takes place among multiple parties,
which do not trust each other. This new scenario has spurred an active
trend of research on safety and security properties for multiparty
interactions. It is often the case that such interactions are
``structured'', i.e. they follow a specified protocol. Since their
introduction in~\cite{CHY08} (as an extension of binary session
calculi), \emph{multiparty session calculi} have been widely used to
model structured communications among multiple parties.
Session calculi are endowed with particular behavioural types called
\emph{session types}, which ensure that communications are not blocked
and follow the expected protocol. Lately, multiparty session calculi have been
enriched with security requirements, 
in order to ensure properties such as access control and leak freedom. 
An account of security analysis in multiparty
session calculi and similar formalisms may be found in the recent
survey~\cite{BCDDGPPTTV}.

A drawback of the existing security-enriched session calculi (such as
those reviewed in~\cite{BCDDGPPTTV}) is that the security requirements
are overly restrictive in some cases. In particular, a party is not
allowed to communicate any kind of public information after receiving
a secret information. This does not seem justified in case the two
pieces of information are totally unrelated. The aim of the present
paper is to overcome this restriction, by designing a type discipline
for a simple multiparty session calculus, which classifies messages
according to their \emph{topics} and allows unrestricted sequencing of
messages on independent topics.  
In this way, we can safely type processes that are
 rejected by previous type systems.

We start by illustrating our approach with a familiar example.

%

\begin{myexample}\label{ex:PC-informal}
  A Programme Committee (PC) discussion may be described as a session
  whose participants are the PC members and whose main topics are the
  submitted papers. All papers are assumed to be unrelated unless they
  share some author. A further topic, unrelated to the papers, is
  constituted by a bibliographic database, which is public but
  possibly not easily accessible to all PC members; hence all PC
    members are allowed to ask other PC members to fetch a document in
    the database for them.  Other topics, unrelated to the previous
  ones, are administrative data of interest to the PC, like email
    addresses.

    At the start of the session, all PC members receive a number of
    papers to review.  During the discussion, PC members receive
    reviews and feedback on the papers in their lot, but possibly also
    on other papers for which they have not declared
    conflict. In this scenario, our typing will ensure 
the following properties:
\begin{myenumerate}
\item A PC member $P_1$ who received confidential information on paper
  $\topic$ can forward this information to another PC member $P_2$ if
  and only $P_2$ is not in conflict with paper $\topic$ nor with any
  related paper; 
%
\item A PC member who received confidential information on some paper
  $\topic$ can subsequently send an email address to any other PC
  member, including
  those in conflict with paper $\topic$; 
\item The PC chair $P_0$ is allowed to request a document
belonging to the bibliographic database to any PC member at any time, even
    after receiving confidential information on some 
paper $\topic$.
This could happen for instance if a PC member $P_1$ in charge of paper
$\topic$ wishes to compare it with a previous paper by a PC member
$P_2$ who is in conflict with paper $\topic$.  Suppose this paper is
in the database but $P_1$ cannot access it; then $P_1$ will express
her concerns about paper $\topic$ to the PC chair $P_0$ and ask him to
retrieve the document from the database.  The point is that $P_0$
himself may not have an easy access to the document; in this case
$P_0$ will forward the request directly to $P_2$.
Intuitively, this should be allowed because the requested document has
the topic $\topicpsi$ of the database,
which is not related to topic $\topic$.
%
\end{myenumerate}
\end{myexample}
In the above example, Property 1 is an access control
(AC) property, which will be handled by assigning to each participant a
reading level for each topic; Property 2 is a leak freedom (LF) 
 property, where the usual ``no write-down'' condition
is relaxed when the topic of the output is independent from that of the
preceding input; finally, Property 3 involves both
AC and LF issues. 
Our type system will ensure a {\em safety} property that is a combination
of AC and of our relaxed LF property.

The next sections present the untyped calculus, the safety definition,
the type system and the main properties of the typed calculus. 


\mysection{Synchronous Multiparty Session Calculus}\label{sec:msc} 
We introduce here our synchronous
multiparty session calculus, which is essentially the LTS version of
the calculus considered in~\cite{DGJPY15}.
\myparagraph{Syntax} 
A multiparty session is an abstraction for
describing multiparty communication protocols~\cite{CHY08}.
It consists of a series of
interactions between a fixed number of participants.

We use the following base sets: \emph{security levels}, ranged over by
$\lev,\levp,\dots$; \emph{topics}, ranged over by
$\topic,\topicpsi,\dots$; \emph{values with levels and topics}, ranged
over by $\val^{\level,\topic},\valu^{\level',\topicpsi},\ldots$;
\emph{expressions}, ranged over by $\e,\e',\ldots$; \emph{expression
  variables}, ranged over by $x,y,z\dots$; \emph{labels}, ranged over
by $\cu,\cuu,\dots$; \emph{session participants}, ranged over by
$\pp,\pq,\ldots$; \emph{process variables}, ranged over by
$X,Y,\dots$; \emph{processes}, ranged over by $P,Q,\dots$; and
\emph{multiparty sessions}, ranged over by $\N,\N',\dots$.

\emph{Processes} $\PP$ are defined by:
\begin{myformula} {\begin{array}{lll}\PP   &   ::=  &
     \procout \q  {\lambda(\e)} \PP  \sep   \procin{\pp}{\lambda(\x)}{\Q}    
     \sep     \cond{\e} \PP  \PP    
            \sep     \PP + \PP                                    
            \sep     \mu X.\PP  
            \sep     X                                           
            \sep     \inact  \end{array}}
\end{myformula}
 The output process $\procout{\q}{\lambda(\e)}{\PP}$ sends
             the value of expression $\e$ with label $\lambda$ to
             participant $\q$.
The input process
$\procin{\pp}{\lambda(\x)}{\Q}$
waits for the value of an expression with label $\lambda$ from participant $\pp$.
The operators of internal and external
choice, denoted $\oplus$ and $+$ respectively, are standard.
We take an equi-recursive view of processes, not distinguishing
between a process $\mu X.\PP$ 
and its unfolding
$\PP\sub{\mu X.\PP}{X}$. 
 We assume that the recursive processes are guarded, i.e. $\mu X.X$ is not a process. 

\smallskip

A \emph{multiparty session} $\N$ is a parallel composition of pairs
(denoted by $\pa\pp\PP$) made of a participant and a process:
\myformulaC{\begin{array}{lll}\N & ::= & \pa\pp\PP \sep \N
    \pc\N \end{array}} 
We will use $\sum\limits_{i\in I} \PP_i$ as
short for $\PP_1+\ldots+\PP_n,$ and $\prod\limits_{i\in I}
\pa{\pp_i}{\PP_i}$ as short for $\pa{\pp_1}{\PP_1}\pc \ldots\pc
\pa{\pp_n}{\PP_n},$ where $I=\set{1,\ldots,n}$.



Security levels and topics, which appear as superscripts of values,
are used to classify values according to two criteria: their degree of
confidentiality and their subject. The use of these two parameters
will become clear in Section~\ref{sec:safety}.

Our calculus is admittedly very simple, since processes are sequential
and thus cannot be involved in more than one session at a time. As a
consequence, it is not necessary to introduce explicit session
channels: within a session, processes are identified as session
participants and can directly communicate with each other, without
ambiguity since the I/O operations mention the communicating partner.


\myparagraphB{Operational semantics} The value $\val^{\level,\topic}$
of an expression $\e$ (notation $\eval\e{\val^{\level,\topic}}$) is
defined as expected, provided that all the values appearing in $e$
have the same topic $\topic$ (this will be guaranteed by our typing)
and the join of their security levels is $\ell$. The semantics of
processes and sessions is given by means of two separate LTS's. The
actions of processes, ranged over by $\vartheta$, are either the
silent action $\tau$ or a visible I/O action $\alpha$ of the form
$\q!\lambda(\val^{\level,\topic})$ or $\pp
?\lambda(\val^{\level,\topic})$.  The actions of sessions, ranged over
by $\kappa$, are either $\tau$ or a \emph{message} of the form
$\comm\pp \val{\level}\lambda{\topic}\q$.

The LTS's for processes and sessions are given by the rules in
Table~\ref{LTS}, defined up to a standard structural congruence
denoted by $\equiv$ (by abuse of notation we use the same symbol for
both processes and sessions), whose definition is in Table \ref{tab:sync:congr}.

 \begin{mytable}{h}
\centerline{$
\begin{array}[t]{@{}c@{}}
\inferrule[\rulename{s-intch 1}]{}{
   \cond{}\PP \Q \equiv \cond{}\Q  \PP
      }
\qquad
\inferrule[\rulename{s-intch 2}]{}{
   \cond{}{(\cond{}\PP \Q ) } \R\equiv \cond{}\PP  {(\cond{}\Q \R)}
      }
\\\\
\inferrule[\rulename{s-extch 1}]{}{
   \PP \external \Q \equiv \Q \external \PP
      }
\qquad
\inferrule[\rulename{s-extch 2}]{}{
   (\PP \external \Q ) \external \R\equiv \PP \external (\Q \external \R)
      }
\\\\
\inferrule[\rulename{s-rec}]{}{
   \mu X.\PP \equiv  \PP\Subst{\mu X.\PP}{X}
      }
\qquad 
\inferrule[\rulename{s-multi}]{}{
   \PP \equiv \Q  \Rightarrow \pa\pp\PP\equiv\pa\pp\Q
      }
\qquad      
  \inferrule[\rulename{s-par 1}]{}{
    \pa\pp{\inact} \pc \N \equiv \N
      }
\\\\
  \inferrule[\rulename{s-par 2}]{}{
    \N \pc \N' \equiv \N' \pc \N
      }
 \qquad
 \inferrule[\rulename{s-par 3}]{}{
    (\N \pc \N') \pc \N'' \equiv \N \pc( \N' \pc \N'')
      }
\end{array}
$}
\caption{Structural congruence.}
\label{tab:sync:congr}
\end{mytable}

\begin{mytable}{}
\[
\begin{array}[ht]{@{}c@{}}
\inferrule[\rulename{r-output}]{
     \eval{\e}{\val^{\level,\topic}}}{
    \procout \q {\lambda(\e)} \PP
    \lts{\q!\lambda(\val^{\level,\topic})}
   \PP}
  \qquad\qquad\qquad
\inferrule[\rulename{r-input}]{}{
   \procin{\pp}{\lambda(\x)}{\Q}
    \lts{\pp ?\lambda(\val^{\level,\topic})}
   \Q\sub{\val^{\level,\topic}}{\x}
   }
  \\ \\
  \inferrule[\rulename{r-int-choice}]{}{
    \cond{\e}{\PP}{\Q}   \lts{\tau} \PP
   }
    \qquad
    \inferrule[\rulename{r-ext-choice}]{\PP\lts{\alpha}\PP'
   }{
    \PP\external\Q  \lts{\alpha} \PP' 
   }
 \qquad
  \inferrule[\rulename{r-struct-proc}]{
  \PP'_1\equiv \PP_1 \quad \PP_1\lts{\vartheta} \PP_2 \quad \PP_2 \equiv \PP'_2
  }
  { 
   \PP'_1 \lts{\vartheta} \PP'_2
  }
  \end{array}
\]
\[
\begin{array}[ht]{@{}c@{}}
\inferrule[\rulename{r-comm}]{
     \PP\lts{\q!\lambda(\val^{\level,\topic})}\PP'\quad \Q\lts{\pp ?\lambda(\val^{\level,\topic})}\Q'}{
    \pa\pp\PP \pc \pa\q \Q
    \lts{\comm\pp \val{\level}\lambda{\topic}\q} 
    \pa\pp{\PP'}\pc\pa\q{\Q'}
    }
  \qquad
  \inferrule[\rulename{r-tau}]{\PP\lts{\tau}\PP'
   }{
    \pa\pp{\PP}  \lts{\tau} \pa\pp\PP'
   }
 \\ \\
 \inferrule[\rulename{r-context}]{
    \N  \lts{\kappa} \N'
  }
  {  \N \pc \N'' \lts{\kappa}  \N' \pc\N''
  }   
  \qquad
  \inferrule[\rulename{r-struct-sess}]{
  \N'_1\equiv \N_1 \quad \N_1\lts{\kappa} \N_2 \quad \N_2 \equiv \N'_2
  }
  { 
   \N'_1 \lts{\kappa} \N'_2
  }
\end{array}
\]
\caption{\label{LTS} LTS rules for processes and sessions.}
\end{mytable}


\mysection{Safety}\label{sec:safety}
Our notion of safety for sessions has two facets: \emph{access
  control} and information flow security or \emph{leak-freedom}. We
assume that security levels $\level, \level'$ form a finite lattice,
ordered by $\levleq$. We denote by $\join$ and $\meet$ the join and
meet operations on the lattice, and by $\bot$ and $\top$ its bottom
and top elements.  The partial ordering $\levleq$ is used to classify
values according to their degree of confidentiality: a value of level
$\bot$ is public, a value of level $\top$ is secret. The ordering also
indicates the authorised direction for information flow: a flow from a
value of level $\lev$ to a value of level $\lev'$ is allowed if and
only if $\lev \levleq \lev'$.

Furthermore, each session participant $\pp$ has a {\em reading
  level} for each topic $\topic$, denoted by $\rho(\pp,\topic)$. In a safe session,
participant $\pp$ will only be able to receive values of level $\level
\levleq \rho(\pp,\topic)$ on topic $\topic$. This requirement assures
  access control.
  
  We also assume an irreflexive and symmetric relation of {\em
      independence} between topics: we denote by
    $\diff\topic\topicpsi$ the fact that $\topic$ and $\topicpsi$ are
    {\em independent} and by $\rel\topic\topicpsi$ (defined as
    $\neg(\diff\topic\topicpsi)$) the fact that $\topic$ and
    $\topicpsi$ are {\em correlated}.  Neither of these two relations
  is transitive in general, as illustrated by
  Example~\ref{ex:PC-informal}, where $\rel{}{}$ is the co-authorship
  relation between papers and $\diff{}{}$ is its complement.
%
%
%

We say that a session is \emph{leak-free}
if, whenever a
participant $\pp$ receives a value of level $\level$ on topic
$\topic$, then $\pp$ can subsequently only send values of level
$\level'\levgeq\level$ on topics related to $\topic$. 
For instance, the output of 
level $\level'$ could be placed within
an internal choice, and this choice could be resolved depending on the
input of level $\level$, since this input is on a related topic.
%
To formalise this requirement we need to look at the \emph{traces} of
multiparty sessions, ranged over by $\sigma, \sigma'$ and defined 
as the sequences of actions that label a transition
sequence. Formally, $\sigma$ is a word on the alphabet containing
$\tau$ 
and the messages
$\comm{\particip}{\val}{{\level}}\lambda{\topic}{\participq}$ for all
participants $\pp, \pq$, labels $\lambda$, values $\val$, security
levels $\level$ and topics $\topic$. Safety is now defined as
  follows, using the 
notion of relay trace: 
\begin{mydefinition}
A \emph{relay trace} is a trace of the form:
\myformula{\sigma \cdot
  \comm{\particip}{\val}{\level}\lambda{\topic}{\participq}\,
 \cdot \sigma' \cdot
 \comm{\participq}{\valu}{{\level'}}{\lambda'}{\topicpsi}{\participr}}
The middle participant $\participq$ is called the \emph{mediator}
between participants $\particip$ and $\participr$. 
\end{mydefinition}

\begin{mydefinition}
\label{def:safety}
A multiparty session $\N$ is {\em safe} if it satisfies: 
\begin{myenumerate}
\item \emph{Access control (AC):} whenever 
$\sigma \cdot
  \comm{\particip}{\val}{{\level}}\lambda{\topic}{\participq}$ 
is a trace of $\N$,
then $\level\levleq\rho(\participq,\topic)$;
\item \emph{Leak freedom (LF):} whenever $\sigma \cdot
  \comm{\particip}{\val}{{\level}}\lambda{\topic}{\participq}\, \cdot
  \sigma' \cdot
  \comm{\participq}{\valu}{{\level'}}{\lambda'}{\topicpsi}{\participr}$
 is a relay trace of $\N$, then either $\level \levleq \level'$ or
  $\diff\topic\topicpsi$.
\end{myenumerate}
\end{mydefinition}
For example the relay trace $\comm \pp \true \top \lambda\topic \pq \cdot
 \comm \pq \false \bot {\lambda'}\topicpsi \pr$ satisfies
the condition of the previous definition if $\rho(\pp,\topic)=\top$ and
$ \diff\topic\topicpsi$.  Intuitively, 
in spite of the ``level drop'' between the two messages, 
their sequencing is harmless because they
  belong to two different conversations.
%
%


 \begin{myexample}\label{ex2}
The PC discussion described in Example \ref{ex:PC-informal}
may be formalised as the session:

\[ \N_{PC} = \prod\limits_{i\in I}
\pa{\pp_i}{\PP_i} \quad
{\rm where} ~ I=\set{1,\ldots,n}.
\] 
Here each participant $\pp_i$ represents a PC member, and $\PP_i$ is
the associated process.  Let us see how the three properties discussed
in Example \ref{ex:PC-informal} can be expressed in $\N_{PC}$.
\begin{enumerate}
\item (AC issue) Here we assume that $\pp_1$ is entitled to receive a
  confidential value $\val^{\level,\topic}$ from some $\pp$. Thus
  $\lev \neq \bot$ and $\level\sleq\rho(\pp_1,\topic)$. Subsequently
  $\pp_1$ forwards this information to $\pp_2$, hence there is a relay
  trace of the form
  $\commsimp{\particip}{\particip_1}{\val}{\level}{\topic}\cdot\commsimp{\particip_1}{\particip_2}{\val}{\level}{\topic}$.
  This trace trivially satisfies LF, and the second message satisfies
  AC if and only if $\level\sleq\rho(\pp_2,\topic)$. Then, if we
  set $\rho(\pp_2,\topic) = \bot\,$ for any $\pp_2$ in conflict with
  $\topic$ and $\rho(\pp_2,\topic) = \top\,$ for any other $\pp_2$,
  Property 1 will be ensured by the safety of $\N_{PC}$.
\item (LF issue) Here the relay trace has the form
$\commsimp{\particip}{\particip_1}{\val}{\lev}{\topic}\cdot\commsimp{\particip_1}{\particip_2}{\val_1}{\bot}{\topic_1}$,
where again $\bot\neq\lev \sleq\rho(\pp_1,\topic)$.
  This trace satisfies LF because 
$\topic_1$ is independent from $\topic$. The second message trivially satisfies
  AC because the email address $\val_1^{\bot,\topic_1}$ has level
  $\bot$ and thus can be read by any participant.

\item (Combination of AC and LF) Here $\pp_1$ sends to the PC Chair
  $\pp_0$ a confidential value $\val_1^{\level,\topic}$, followed by a
  request for a public document of topic $\topicpsi$, and then waits to
  receive this document from $\pp_0$. The behaviour of $\pp_0$ is
  dual for the first two steps, but then $\pp_0$ asks
  $\pp_2$, who is in conflict with paper $\topic$, to fetch the
  document for him, before sending it back to $\pp_1$.

Processes implementing the behaviour of the PC Chair and of the involved PC members are:
\myformula{\begin{array}{lll}\PP_0 &=&
  \pp_1?(x).\pp_1?(y).\pp_2!(y).\pp_2?(z).\pp_1!(z).\inact\\
 \PP_1 &=&
  \pp_0!(\val_1^{\level,\topic}). \pp_0!(\val_2^{\bot,\topicpsi}).\pp_0?(x).\inact\\
   \PP_2 &=&
  \pp_0?(x).\pp_0!(\val_3^{\bot,\topicpsi}).\inact
  \end{array}}
%


Intuitively, the reading levels of $\pp_0, \pp_1$ and $\pp_2$ should
be $\rho(\pp_0,\topic) = \rho(\pp_0,\topicpsi) = \top$,
$\rho(\pp_1,\topic) = \lev$, $\rho(\pp_1,\topicpsi) = \bot$, and
$\rho(\pp_2,\topic) = \rho(\pp_2,\topicpsi) = \bot$.
%
%
Consider now the following trace of session $\N_{PC}$: 
\[\sigma =
\commsimp{\particip_1}{\particip_0}{\val_1}{\level}{\topic}\cdot
\commsimp{\particip_1}{\particip_0}{\val_2}{\bot}{\topicpsi}\cdot
\commsimp{\particip_0}{\particip_2}{\val_2}{\bot}{\topicpsi}\cdot
\commsimp{\particip_2}{\particip_0}{\val_3}{\bot}{\topicpsi}\cdot
\commsimp{\particip_0}{\particip_1}{\val_3}{\bot}{\topicpsi}\]
With the above reading levels, each message in trace $\sigma$
satisfies AC.
Moreover, trace $\sigma$ contains three relay traces: 
$\commsimp{\particip_1}{\particip_0}{\val_2}{\bot}{\topicpsi}\cdot\commsimp{\particip_0}{\particip_2}{\val_2}{\bot}{\topicpsi}$,
$\commsimp{\particip_0}{\particip_2}{\val_2}{\bot}{\topicpsi}\cdot
\commsimp{\particip_2}{\particip_0}{\val_3}{\bot}{\topicpsi}$,
and $\commsimp{\particip_2}{\particip_0}{\val_3}{\bot}{\topicpsi}\cdot
\commsimp{\particip_0}{\particip_1}{\val_3}{\bot}{\topicpsi}$,
which trivially satisfy LF since all values have
level $\bot$. 


\end{enumerate}

\end{myexample}

\mysection{Type System}\label{sec:ts} Our type system enriches the
system of~\cite{DGJPY15} with security levels and topics.  

\myparagraph{Types} {\em Sorts}  are ranged over by $\SOT$ and defined by:\qquad
$ \SOT \quad ::= \quad \tnat \sep \tint \sep\tbool\sep\String$
\\
{\em Global types} describe the whole conversation scenarios of
multiparty sessions. They are generated by:
\begin{myformula}{
\begin{tabular}{rclrrrclr}
 $\G$ & $::=$ & $\Gvti {\level} {\topic} \pp \q \lambda \SOT {\G}$  & \sep &  $\mu\ty. \G$ & \sep &$\ty$   &\sep &$\tend$
\end{tabular}}\end{myformula}
{\em Session types} correspond to the views of the individual
participants. 
They can be either unions of outputs or intersections of inputs. 
The grammar of session types, ranged over by $\T$, is then 
\begin{myformula}{\T ::= 
           \bigvee_{i\in I}\tout \q{\lambda_i}{\SortI{\SOT}{i}{\level}{\topic}}.\T_i \sep\bigwedge_{i\in I}\tin\pp{\lambda_i}{\SortI{\SOT}{i}{\level}{\topic}}.\T_i \sep
           \mu\ty. \T \sep \ty \sep \tend
}\end{myformula}
We require  that $\lambda_i\not=\lambda_j$ with $i\not=j$ and $i,j\in I$. 

\smallskip

We give now conditions on session types which will guarantee session safety. 

\begin{mydefinition}
A pair of a security level $\level$ and a topic $\topic$ {\em agrees}
with a session type $\T$ (notation $ \pair\level\topic\prec\T$) if 
$\T$ specifies that only values of level $\level'\sqsupseteq\level$ are sent on topics related with $\topic$:
\myformula{\begin{array}{@{}c@{}}
\inferrule[\rulename{agr-end}]{}
  {\pair\level\topic\prec\tend}\qquad
\cinferrule[\rulename{agr-out}]{
 \forall i\in I: \pair\level\topic\prec\T_i \text{ (either } \level\levleq\lev_i' \text{ or } \diff\topic\topicpsi_i)
}{\pair\level\topic\prec\bigvee_{i\in I}
 \tout\q{\lambda_i}{\SOT^{\level_i',\topicpsi_i}}.\T_i
 }
\qquad
\cinferrule[\rulename{agr-in}]{
\forall i\in I: \pair\level\topic\prec \T_i  }{
  \pair\level\topic\prec\bigwedge_{i\in I}\tin\pp{\lambda_i}{\SortI{\SOT}{i}{\level}{\topic}}.\T_i
}
\end{array}}
\end{mydefinition}
\begin{mydefinition}
A closed session type $\T$ is a {\em safe session type} if $\safe\T$
can be derived from the rules:
 \myformula{
\begin{array}{@{}c@{}}
\inferrule[\rulename{safe-end}]{}
  {\safe\tend}\qquad
\cinferrule[\rulename{safe-out}]{
 \forall i\in I: \quad\safe{\T_i}\quad \level_i\levleq\rho(\q,\topic_i)
}{\safe{\bigvee_{i\in I}
 \tout\q{\lambda_i}{\SortI{\SOT}{i}{\level}{\topic}}.\T_i}
 }
\qquad
\cinferrule[\rulename{safe-in}]{
\forall i\in I: \quad\safe{\T_i}\quad\pair{\level_i}{\topic_i}\prec \T_i  }{
  \safe{\bigwedge_{i\in I}\tin\pp{\lambda_i}{\SortI{\SOT}{i}{\level}{\topic}}.\T_i}
}
\end{array}
}
\end{mydefinition}
The double line in the above rules means that they are {\em coinductive}~\cite[21.1]{pierce}. This is necessary since session types are recursive and under the equi-recursive approach 
the types in the premises can coincide with the types in the conclusion. For example $\pp?\lambda(\tbool^{\top,\topic}).\pr!\lambda'( \tbool^{\bot,\topicpsi}).\tend$ is a safe type if $\rho(\pp,\topic)=\top$ and $ \diff\topic\topicpsi$. 

We only allow safe types in the typing rules for processes and
multiparty sessions. As will be established in Theorem~\ref{thm:type-safe}, the conditions
in rules \rulename{safe-out} and \rulename{safe-in} of safe session types
assure respectively access control and leak freedom, namely 
Properties 1 and 2 of session safety (Definition~\ref{def:safety}).
\myparagraph{Typing rules} We distinguish three
kinds of typing judgments. Expressions are typed by sorts with levels
and topics, processes are typed by session types and multiparty
sessions are typed by global types:
\begin{myformula}{
  \dere\Gamma\e{\SOT^{\level,\topic}} \qquad\qquad \derS \Gamma\PP\T\Delta \qquad\qquad \ders\N\G
}\end{myformula}
Here $\Gamma$ is the {\em environment} that associates expression
variables with sorts (decorated by levels and topics) and process
variables with safe session types: $\Gamma ::= \emptyset \sep \Gamma,
x:\SOT^{\level,\topic} \sep \Gamma, X:\T$.

The typing rules for expressions in Table \ref{tab:exp-type} are
almost standard, but for the treatment of topics. A value of level
$\level$ and topic $\topic$ is typed with the appropriate sort
type decorated by $\level$ and $\topic$.  Expressions cannot contain
subexpressions of different topics.  This limitation could be easily
overcome by allowing sets of topics. In this way we could associate to an expression the set of topics of its subexpressions. 
The sets of topics would naturally build 
  a lattice, where the order is given by subset inclusion.

\begin{mytable}{h}
\centerline{$
\begin{array}{@{}c@{}}
\inferrule[\rulename{exp-var}]{}
  {\Gamma, \x:\SOT^{\level,\topic} \vdash \x:\SOT^{\level,\topic} }
 \qquad
\inferrule[\rulename{exp-val}]
{}
{
  \Gamma \vdash \val^{\level,\topic}:\SOT^{\level,\topic}
}
\qquad
\inferrule[\rulename{exp-op}]{
 \Gamma \vdash \e_1: \SOT_1^{\level_1,\topic} 
 \quad
 \Gamma \vdash \e_2: \SOT_2^{\level_2,\topic}
 \quad
 \op: \SOT_1,\SOT_2 \rightarrow \SOT_3
}{
  \Gamma \vdash \e_1 \op \e_2: \SOT_3^{\level_1 \sqcup \level_2,\topic}
 }
\end{array}
 $}
  \caption{\label{tab:exp-type} Typing rules for expressions.}
\end{mytable}

\begin{mytable}{}
\centerline{$
\begin{array}{c}
\myrule{\dere\Gamma\e{\SOT^{\level,\topic}}~~\  \derS {\Gamma}{\PP}\T\Delta
  }{\derS \Gamma{\procout  \q   {\lambda(\e)}
      \PP}{\q!\lambda(\SOT^{\level,\topic}).\T}{\Delta\sqcap\set{\topic:\level}}}{[t-out]}
  \qquad
  \myrule{\derS {\Gamma,x:\SOT^{\level,\topic}}{\Q}\T\Delta
  }{\derS \Gamma{{\procin  \pp   {\lambda(\x)} \Q}}{\pp?\lambda(\SOT^{\level,\topic}).\T}{\Delta}}{[t-in]}
   \\\\
   \myrule{
  \derS {\Gamma}{\PP_1}{\T_1}{\Delta_1}~~\derS {\Gamma}{\PP_2}{\T_2}{\Delta_2}}{\derS \Gamma{{\cond{\e}{\PP_1}{\PP_2}}}{\T_1\vee \T_2}{\Delta_1\sqcap\Delta_2}}{[t-i-choice]}
   \qquad  
    \myrule{\derS {\Gamma}{\PP_1}{\T_1}{\Delta_1}~~\derS {\Gamma}{\PP_2}{\T_2}{\Delta_2}}{\derS \Gamma{{{\PP_1}+{\PP_2}}}{\T_1\wedge \T_2}{\Delta_1\sqcap\Delta_2}}{[t-e-choice]}
  \\\\
  \myrule{\derS {\Gamma,X:\T}{\PP}{\T}\Delta}
  {\derS \Gamma{\mu X.\PP}\T\Delta}{[t-rec]}
  \quad
   \derS {\Gamma,X:\T}{X}{\T}\emptyset~~\rln{[t-var]}  
    \qquad
   \derS \Gamma{ \inact}\tend\emptyset~~\rln{[t-$\inact$]}
    \end{array}
$}
\caption{\label{Ntab:sync:typing} Typing rules for processes.}
\end{mytable}

Processes have the expected types. Let us note that the syntax of
session types only allows output processes in internal choices (typed by
unions) and input processes in external choices (typed
by intersections). Table~\ref{Ntab:sync:typing} gives the typing rules for
processes. For example, if $\rho(\pp,\topic)=\top$ and $
\diff\topic\topicpsi$ we can derive $\derS{}
{\procin{\pp}{\lambda(\x)}{\procout
    \pr{\lambda'(\false^{\bot,\topicpsi})}{\inact} }}
{\pp?\lambda(\tbool^{\top,\topic}).\pr!
  \lambda'(\tbool^{\bot,\topicpsi}).\tend}{} $, while this process
cannot be typed otherwise. 
Notice that the process obtained by erasing topics is not typable in
the system of~\cite{CCD14}, where the typing rule for input
  requires that the level of the input be lower than or
  equal to the level of the following output. Similarly,
in the monitored
  semantics of~\cite{CCD16}, this input would raise the monitor level to
  $\top$ and then the monitor would produce an error when applied to the output of level $\bot$.

\smallskip

A session is typable when its parallel components
can play as participants of a whole communication protocol or they are terminated.
To formalise this we need some definitions. 

The {\em subtyping} relation $\leq$ between session types as defined in Table \ref{tab:sync:subt} is simply the set-theoretic inclusion between intersections and unions. 
The double line in these rules means that subtyping is co-inductively defined.

\begin{mytable}{h}
\centerline{$
\begin{array}{@{}c@{}}
\inferrule[\rulename{sub-end}]{}
  {\tend \subt \tend}\qquad
\cinferrule[\rulename{sub-in}]{
\forall i\in I: \quad \T_i \subt\T_i' }{
  \bigwedge_{i\in I\cup J} \tin\pp{\lambda_i}{\SortI{\SOT}{i}{\level}{\topic}}.\T_i \subt \bigwedge_{i\in I}\tin\pp{\lambda_i}{\SortI{\SOT}{i}{\level}{\topic}}.\T_i'
}
\qquad
\cinferrule[\rulename{sub-out}]{
 \forall i\in I: \quad \T_i \subt \T'_i  
}{\bigvee_{i\in I}
  \tout\pp{\lambda_i}{\SortI{\SOT}{i}{\level}{\topic}}.\T_i
  \subt\bigvee_{i\in I\cup J}
 \tout\pp{\lambda_i}{\SortI{\SOT}{i}{\level}{\topic}}.\T_i'
 }
\end{array}
$}
\caption{\label{tab:sync:subt} Subtyping rules.}
\end{mytable}

The {\em projection} of the global type $\G$ 
on participant $\pp$, notation $\proj \G \pp$, is as usual~\cite{CHY08}, and it is reported in Table \ref{tab:Npro}. 
We shall consider projectable global types only§.

\begin{mytable}{h} 
\centerline{$\begin{array}{c}\proj {\Gvti {\level} {\topic} \pp\q \lambda \SOT {\G}} \pr=\begin{cases}
\bigvee_{i\in I}\tout\q{\lambda_i}{\SortI{\SOT}{i}{\level}{\topic}}. \proj{\G_i}\pr     & \text{if }\pr=\pp, \\
\bigwedge_{i\in I}\tin\pp{\lambda_i}{\SortI{\SOT}{i}{\level}{\topic}}. \proj{\G_i}\pr     & \text{if }\pr=\q, \\
  \proj{\G_i}\pr    & \text{if $\pr\not=\pp$, $\pr\not=\q$}\\
  & \text{and $\proj{\G_i}\pr=\proj{\G_j}\pr$ for all $i,j\in I$}.
\end{cases}\\[7mm]
\proj {\mu\ty.\G}\pr=\begin{cases}
 \proj{\G}\pr     & \text{if $\pr$ occurs in $\G$}, \\
 \tend     & \text{otherwise}.
\end{cases}\qquad\qquad \proj{\ty}\pr=\ty\qquad\qquad \proj{\tend}\pr=\tend
\end{array}
$}
\caption{Projection of global types onto participants.}
\label{tab:Npro}
\end{mytable}

We define the set $\PART\G$ of participants of a global type $\G$ as expected: 
\begin{myformula}{
\begin{array}{c}
 \PART{\Gvti {\level} {\topic} \pp\q \lambda \SOT {\G}}=\{\pp,\q\}\cup\PART{\G_i} ~(i\in I)\footnotemark\\
 \PART{\mu\ty. \G}=\PART{\G} 
\qquad
  \PART\ty=\emptyset \qquad \PART\tend=\emptyset \end{array}
}\end{myformula}
 \footnotetext{The projectability of $\G$ assures $\PART{\G_i}=\PART{\G_j}$ for all $i,j\in I$.}

 We can now explain the typing rule for sessions:
 \myformula{\myrule{\forall i \in \{1, \ldots, n \}: \quad\derS {}{\PP_i}{\T_i}{} \qquad\T_i\leq\proj\G
     {\pp_i}\quad 
\quad \PART \G\subseteq\{\pp_1,\ldots,\pp_n\} } {\ders{\pa
       {\pp_1}\PP_1\pc\ldots|\,\pa{\pp_n}\PP_n }{\G} }{[t-sess]}} 
Note that all $\pp_i$ must be distinct, since
 the premise assumes $\{\pp_1,\ldots,\pp_n\}$ to be a set of
  $n$ elements.
 The condition $\T_i\leq\proj\G {\pp_i}$ assures that the type of the
 process paired with participant $\pp_i$ is ``better'' than the
 projection of the global type $\G$ on $\pp_i$.  The inclusion of
 $\PART \G$ in the set $\{\pp_1,\ldots,\pp_n\}$ allows 
 sessions containing $\pa\pp\inact$ to be typed, a property needed to assure
 invariance of types under structural congruence.
 
 \begin{myexampleB}
The communication protocol described in Examples~\ref{ex:PC-informal} and ~\ref{ex2}, Item 3,
can be formalised (omitting labels) by the global type:
\myformula{\pp_1\to\pp_0:\str^{\level,\topic}.\pp_1\to\pp_0:\str^{\bot,\topicpsi}.\pp_0\to\pp_2:\str^{\bot,\topicpsi}.\pp_2\to\pp_0:\str^{\bot,\topicpsi}.\pp_0\to\pp_1:\str^{\bot,\topicpsi}.\End}
where $\str$ is short for $\String$.

The session type of the PC chair $\p_0$ is then:
\myformula{\pp_1?\str^{\level,\topic}.\pp_1?\str^{\bot,\topicpsi}.\pp_2!\str^{\bot,\topicpsi}.\pp_2?\str^{\bot,\topicpsi}.\pp_1!\str^{\bot,\topicpsi}.\End}
This type is safe, 
since $\topic$ and $\topicpsi$ are unrelated. In fact we can check
that
\myformula{\pair\level\topic\prec\pp_1?\str^{\bot,\topicpsi}.\pp_2!\str^{\bot,\topicpsi}.\pp_2?\str^{\bot,\topicpsi}.\pp_1!\str^{\bot,\topicpsi}.\End.}

\end{myexampleB}

\mysection{Main Properties}\label{sec:results}

The basic soundness property of the typing system w.r.t. operational
semantics is subject reduction. As usual with types expressing
communications, the reduction of sessions ``consumes'' the types. This
consumption can be formalised by means of a reduction. In our system
we need to reduce both session types and global types.

The reduction of session types is the smallest pre-order relation closed under the rules:
\begin{myformula}{\begin{array}{lll}
\T\vee\T'\RedT\T&\qquad \pp!\lambda(S^{\level,\topic}).\T\RedT\T&\qquad\bigwedge_{i\in I}\tin\pp{\lambda_i}{\SortI{\SOT}{i}{\level}{\topic}}.\T_i \RedT\T_i
\end{array}}\end{myformula}
These rules mimic respectively internal choice, output and external
choice among inputs.

\begin{mytable}{b} 
\centerline{$\begin{array}{c}\redG{(\Gvti {\level} {\topic} \participr \particips \lambda \SOT {\G})}\pp\lambda\q=\begin{cases}
  \G_{i_0}    & \text{if }\participr=\pp,\\
  & \particips=\q,\\
  & \lambda_{i_0}=\lambda~~i_0 \in I\\
   \Gvtir {\level} {\topic} \participr \particips \lambda \SOT {\G}
   & \text{otherwise}
\end{cases}\\\\
\redG{(\mu\ty.\G)}\pp\lambda\q= \mu\ty.\redG{\G}\pp\lambda\q
\end{array}$}
\caption{Residual after a communication.}
\label{tab:eg}
\end{mytable}

The reduction of global types is the smallest pre-order relation closed under the rule:
\begin{myformula}{\G\RedG\redG\G\pp\lambda\q}\end{myformula}
where $\redG\G\pp\lambda\q$ is the global type obtained from $\G$ by
executing the communication $\pp\lts{\lambda}\q$. We dub
$\redG\G\pp\lambda\q$ the {\em residual} after the communication
$\pp\lts{\lambda}\q$ in the global type $\G$, whose definition is given in
Table \ref{tab:eg}.  Notice that $\redG\G\pp\lambda\q$ is defined only
if $\pp\lts{\lambda}\q$ occurs in $\G$, since both
$\redG\tend\pp\lambda\q$ and $\redG\ty\pp\lambda\q$ are undefined.
For example, if
$\G=\pr\to\s:\lambda'(\tnat^{\bot,\topic}).\pp\to\q:\lambda(\tbool^{\top,\topicpsi})$,
then
\myformula{\redG\G\pp\lambda\q=\pr\to\s:\lambda'(\tnat^{\bot,\topic}).}
The reduction rule for global types is more
  involved
  than that for session types, since the global types do not prescribe
  an order on communications between disjoint pairs of participants.

\smallskip

We can now show that session reduction transforms the global type of a
session into its residual, and the session types of the processes into
their reducts.  Besides substitution and inversion lemmas, the proof
of subject reduction is based on the relations between subtyping,
projection and erasure of communications.
\begin{lemma}\label{lem:substitution}
If $\Gamma \vdash \e:\SOT^{\level,\topic}$ and $\e\downarrow\val^{\level,\topic}$ and
	$\derS{\Gamma, \x:\SOT^{\level,\topic}}{\PP}{\T}{}$, then
	$\derS{\Gamma}{\PP\Subst{\val^{\level,\topic}}{\x}}{\T}{}$.
\end{lemma}

\begin{proof} Standard. \end{proof}

\begin{mylemma}\label{lem:generation}
\hfill
\begin{myenumerate}
\item\label{lem:generation2} If $\derS{\Gamma}{\procout \pp  {\lambda(\e)} \PP} {\T}{}$, then $\T = \pp !\lambda(\SOT^{\level,\topic}).\T'$ and
	$\Gamma \vdash \e: \SOT^{\level,\topic}$ and $\derS{\Gamma}{\PP} {\T'}{}$.
	\item\label{lem:generation1} If $\derS{\Gamma}{\procin{\pp}{\lambda(\x)}{\PP}}{\T}{}$, then $\T = \pp ? \lambda(\SOT^{\level,\topic}).\T'$
	and $\derS{\Gamma, \x: \SOT^{\level,\topic}}{\PP}{\T'}{}$.
\item\label{lem:generation4} If $\derS{\Gamma}{\PP \oplus \Q}{\T}{}$, then $\T = \T_1 \vee \T_2$ and $\derS{\Gamma}{\PP}{\T_1}{}$
	and $\derS{\Gamma}{\Q}{\T_2}{}$.	
\item\label{lem:generation3} If $\derS{\Gamma}{\PP + \Q}{\T}{}$, then $\T = \T_1 \wedge \T_2$ and $\derS{\Gamma}{\PP}{\T_1}{}$
	and $\derS{\Gamma}{\Q}{\T_2}{}$.
\item\label{lem:generation5} If $\ders{\pa {\pp_1}\PP_1\pc\ldots|\,\pa{\pp_n}\PP_n
  }{\G}$, then $\derS {}{\PP_i}{\T_i}{}$ and $\T_i\leq\proj\G {\pp_i}$ for $1\leq i\leq n$ and
   $\PART \G\subseteq\{\pp_1,\ldots,\pp_n\}$.
\end{myenumerate}
\end{mylemma}
\begin{proof} By observing that the type assignment system for processes and multiparty sessions is syntax directed. \end{proof}

\begin{mylemma}\label{lem:erase}
 If $\tout\q\lambda{\SOT^{\level,\topic}}.{\T}\leq\proj\G\pp$ and $\tin\pp\lambda{\SOT^{\level,\topic}}.{\T'}\wedge\T''\leq\proj\G\q$, then $\T\leq\proj{(\redG\G\pp\lambda\q)}\pp$ and $\T'\leq\proj{(\redG\G\pp\lambda\q)}\q$. Moreover $\proj\G\participr=\proj{(\redG\G\pp\lambda\q)}\participr\,$ for $\participr\not=\pp$, $\participr\not=\q$.
\end{mylemma}
\begin{proof} By induction on $\G$ and by cases on the definition of $\redG\G\pp\lambda\q$.  Notice that    $\G$ can only be $\Gvti {\level} {\topic} {\particips_1} {\particips_2} \lambda \SOT {\G}$ with either $\particips_1=\pp$ and $\particips_2=\pq$ or $\set{\particips_1,\particips_2}\cap\set{\pp,\q}=\emptyset$, since otherwise the types in the statement of the lemma could not be subtypes of the given projections of $\G$. 

 If $\G=\Gvti {\level} {\topic} \pp \q \lambda \SOT {\G}$, then $\proj\G\pp=\bigvee_{i\in I}\tout \q{\lambda_i}{\SortI{\SOT}{i}{\level}{\topic}}.\proj{\G_i}\pp$  and 
            \mbox{$\proj\G\q=\bigwedge_{i\in I}\tin\pp{\lambda_i}{\SortI{\SOT}{i}{\level}{\topic}}.\proj{\G_i}\q$.} From $\tout\q\lambda{\SOT^{\level,\topic}}.{\T}\leq\bigvee_{i\in I}\tout \q{\lambda_i}{\SortI{\SOT}{i}{\level}{\topic}}.\proj{\G_i}\pp$ we get $\lambda=\lambda_{i_0}$ and $\T\leq\proj{\G_{i_0}}\pp$ for some $i_0\in I$. From $\tin\pp\lambda{\SOT^{\level,\topic}}.{\T'}\wedge\T''\leq\bigwedge_{i\in I}\tin\pp{\lambda_i}{\SortI{\SOT}{i}{\level}{\topic}}.\proj{\G_i}\q$ and $\lambda=\lambda_{i_0}$ we get $\T'\leq\proj{\G_{i_0}}\q$. This implies \myformula{\T\leq\proj{(\redG\G\pp\lambda\q)}\pp\text{ and }\T'\leq\proj{(\redG\G\pp\lambda\q)}\q,} since $\proj{(\redG\G\pp\lambda\q)}\pp=\proj{\G_{i_0}}\pp$ and $\proj{(\redG\G\pp\lambda\q)}\q=\proj{\G_{i_0}}\q$. If $\participr\not=\pp$, $\participr\not=\q$, then by definition of projection $\proj\G\participr=\proj{\G_{i_0}}\participr$ for an arbitrary $i_0\in I$, and then $\proj\G\participr=\proj{(\redG\G\pp\lambda\q)}\participr$ by definition of residual.

            If $\G=\Gvti {\level} {\topic} {\particips_1}
            {\particips_2} \lambda \SOT {\G}$ and
            $\set{\particips_1,\particips_2}\cap\set{\pp,\q}=\emptyset$,
            then $\proj\G\pp=\proj{\G_{i_0}}\pp$ and
            $\proj\G\q=\proj{\G_{i_0}}\q$ for an arbitrary $i_0\in
            I$. By definition of residual
            \myformula{\redG\G\pp\lambda\q=\Gvtir {\level} {\topic}
              {\particips_1}{\particips_2 }\lambda \SOT {\G},} which
            implies
            $\proj{(\redG\G\pp\lambda\q)}\pp=\proj{(\redG{\G_{i_0}}\pp\lambda\q)}\pp$
            and
            $\proj{(\redG\G\pp\lambda\q)}\q=\proj{(\redG{\G_{i_0}}\pp\lambda\q)}\q$.\\ Notice
            that the choice of $i_0$ does not modify the projection,
            by definition of projectability.  We get
            $\tout\q\lambda{\SOT^{\level,\topic}}.{\T}\leq\proj{\G_{i_0}}\pp$
            and
            $\tin\pp\lambda{\SOT^{\level,\topic}}.{\T'}\wedge\T''\leq\proj{\G_{i_0}}\q$,
            which imply by induction
            \mbox{$\T\leq\proj{(\redG{\G_{i_0}}\pp\lambda\q)}\pp$} and
            $\T'\leq\proj{(\redG{\G_{i_0}}\pp\lambda\q)}\q$.\\ 
\indent            
Let $\participr\not=\pp$, $\participr\not=\q$.\\ 
If $\participr=\particips_1$, then $\proj\G\participr=\bigvee_{i\in
  I}\tout
{\particips_2}{\lambda_i}{\SortI{\SOT}{i}{\level}{\topic}}.\proj{\G_i}\participr$
and \myformula{\proj{(\redG\G\pp\lambda\q)}\participr=\bigvee_{i\in
    I}\tout
  {\particips_2}{\lambda_i}{\SortI{\SOT}{i}{\level}{\topic}}.\proj{(\redG{\G_i}\pp\lambda\q)}\participr,}
so we may conclude, since by induction $\proj{\G_i}\participr=\proj{(\redG{\G_i}\pp\lambda\q)}\participr$ for all $i\in I$.\\
If $\participr=\particips_2$, then $\proj\G\participr=\bigwedge_{i\in
  I}\tin
{\particips_1}{\lambda_i}{\SortI{\SOT}{i}{\level}{\topic}}.\proj{\G_i}\participr$
and \myformula{\proj{(\redG\G\pp\lambda\q)}\participr=\bigwedge_{i\in
    I}\tin
  {\particips_1}{\lambda_i}{\SortI{\SOT}{i}{\level}{\topic}}.\proj{(\redG{\G_i}\pp\lambda\q)}\participr,}
so we may conclude using induction as in the previous case.\\
If $\participr\not\in\set{\particips_1,\particips_2}$, then
$\proj\G\participr=\proj{\G_{i_0}}\participr$ and
$\proj{(\redG\G\pp\lambda\q)}\participr=\proj{(\redG{\G_{i_0}}\pp\lambda\q)}\participr$
for an arbitrary $i_0\in I$. Again, we can conclude using induction.
\end{proof}

\begin{mytheorem}{Subject reduction}\label{thm:SR}
If $\,\pa {\pp}\PP\pc\N \lts{\kappa}\pa {\pp}\PP'\pc \N'$, ~
$\ders{\pa {\pp}\PP\pc\N}\G$ and $\derS {}{\PP}{\T}{}$, then:
\begin{myenumerate}
\item $\ders{\pa {\pp}\PP'\pc\N'}\G'$ for some $\G'$ such that $\G\RedG^*\G'$;
\item $\derS {}{\PP'}{\T'}{}$ for some $\T'$ such that $\T\RedT^*\T'$.
\end{myenumerate}
\end{mytheorem}

\begin{proof}
We only consider the more interesting reduction, i.e., when $\PP$ is reduced. 
We distinguish three cases according to the shape of
$\kappa$. 

\smallskip

{\em Case $\kappa = \tau$:} then $\PP\equiv\PP_1\oplus\PP_2$ and $\PP'\equiv\PP_1$ and $\N'\equiv\N$. By Lemma~\ref{lem:generation}(\ref{lem:generation5}) and (\ref{lem:generation4}) $\T\leq\proj\G\pp$ and $\T=\T_1\vee\T_2$ and 	$\derS {}{\PP_1}{\T_1}{}$. We can then choose $\G'=\G$ and $\T'=\T_1$. 

\smallskip

{\em  Case $\kappa = \comm\pp \val{\level}\lambda{\topic}\q$:} then $\PP\equiv\procout \q  {\lambda(\e)} {\PP'}$ and $\N\equiv\pa\q\procin{\pp}{\lambda(\x)}{\Q_1}+\Q_2\pc\N''$ and \myformula{\N'\equiv\pa\q{\Q_1\sub{\val^{\level,\topic}}{\x}}\pc\N'',} where $\eval{\e}{\val^{\level,\topic}}$. By Lemma~\ref{lem:generation}(\ref{lem:generation5}) and (\ref{lem:generation2}) $\T\leq\proj\G\pp$ and $\T=\q !\lambda(\SOT^{\level,\topic}).\T'$ and
	$ \vdash \e: \SOT^{\level,\topic}$ and \mbox{$\derS{}{\PP'} {\T'}{}$.} By Lemma~\ref{lem:generation}(\ref{lem:generation5}) and (\ref{lem:generation3}) and (\ref{lem:generation1}) \mbox{$\T_1\wedge\T_2\leq\proj\G\q$} and  $\derS{}{\procin{\pp}{\lambda(\x)}{\Q_1}}{\T_1}{}$ and $\derS{}{\Q_2}{\T_2}{}$ and $\T_1 = \pp ? \lambda(\SOT_1^{\level',\topicpsi}).\T_1'$
	and $\derS{\x: \SOT_1^{\level',\topicpsi}}{\Q_1}{\T'_1}{}$. From $\q !\lambda(\SOT^{\level,\topic}).\T'\leq\proj\G\pp$ and $\pp ? \lambda(\SOT_1^{\level',\topicpsi}).\T_1'\wedge\T_2\leq\proj\G\q$ we get $\SOT=\SOT_1$ and $\level=\level'$ and $\topic=\topicpsi$. 
	By Lemma~\ref{lem:substitution} $ \vdash \e: \SOT^{\level,\topic}$ and $\derS{\x: \SOT^{\level,\topic}}{\Q_1}{\T'_1}{}$ imply 
	$\derS{}{\Q_1\sub{\val^{\level,\topic}}{\x}}{\T'_1}{}$. Then we choose $\G'=\redG\G\pp\lambda\q$, since Lemma~\ref{lem:erase} gives $\T'\leq\proj{(\redG\G\pp\lambda\q)}\pp$ and $\T_1'\leq\proj{(\redG\G\pp\lambda\q)}\q$ and the same projections for all other participants of $\G$.

\smallskip
	
{\em Case $\kappa = \comm\q \val{\level}\lambda{\topic}\pp$:} then $\PP\equiv\procin{\q}{\lambda(\x)}{\PP_1}+\PP_2$  and $\N\equiv\pa\q\procout \pp  {\lambda(\e)} {\Q}\pc\N''$ and $\PP'=\PP_1\sub{\val^{\level,\topic}}{\x}$ and $\N'\equiv\pa\q{\Q}\pc\N''$, where $\eval{\e}{\val^{\level,\topic}}$. By Lemma~\ref{lem:generation}(\ref{lem:generation5}) and (\ref{lem:generation3}) and (\ref{lem:generation1}) $\T=\T_1\wedge\T_2\leq\proj\G\pp$ and  \mbox{$\derS{}{\procin{\q}{\lambda(\x)}{\PP_1}}{\T_1}{}$} and $\derS{}{\PP_2}{\T_2}{}$ and $\T_1 = \q ? \lambda(\SOT^{\level,\topic}).\T'$
	and $\derS{\x: \SOT^{\level,\topic}}{\PP_1}{\T'}{}$. 	
	By Lemma~\ref{lem:generation}(\ref{lem:generation5}) and (\ref{lem:generation2}) $\T_3\leq\proj\G\q$ and $\derS{}{\procout \pp  {\lambda(\e)} {\Q}}{\T_3}{}$ and $\T_3=\pp !\lambda(\SOT_1^{\level',\topicpsi}).\T_3'$ and
	$ \vdash \e: \SOT_1^{\level',\topicpsi}$ and $\derS{}{\Q} {\T_3'}{}$. From \mbox{$\q ? \lambda(\SOT^{\level,\topic}).\T'\wedge\T_2\leq\proj\G\pp$} and $\pp !\lambda(\SOT_1^{\level',\topicpsi}).\T_3'\leq\proj\G\q$ we get $\SOT=\SOT_1$ and $\level=\level'$ and $\topic=\topicpsi$. 
	By Lemma~\ref{lem:substitution} $ \vdash \e: \SOT^{\level,\topic}$ and $\derS{\x: \SOT^{\level,\topic}}{\PP_1}{\T'}{}$ imply 
	$\derS{}{\PP_1\sub{\val^{\level,\topic}}{\x}}{\T'}{}$. Then we
        take $\G'=\redG\G\q\lambda\pp$, since Lemma~\ref{lem:erase}
        gives $\T'\leq\proj{(\redG\G\pp\lambda\q)}\pp$ and
        $\T_3'\leq\proj{(\redG\G\pp\lambda\q)}\q$  and the same
        projections for all other participants of $\G$. 
\end{proof}

We may now prove our main result:

\begin{mytheorem}{Soundness}\label{thm:type-safe}
If $\N$ is typable,  then $\N$ is safe.
\end{mytheorem}

\begin{proof}
Suppose that $\N$ is safely typed.  If $\N$ generates the trace
$\sigma \cdot\comm\pp \val{\level}\lambda{\topic}\q$, then
\myformula{\N \lts{~~~\sigma~~~}\pa\pp{\PP }\pc\pa\pq{\Q}
  \pc\N'\lts{\comm\pp \val{\level}\lambda{\topic}\q}
  \pa\pp{\PP'}\pc\pa\pq{\Q'}\pc \N'.}  From $\pa\pp{\PP }
\pc\pa\pq{\Q}\pc\N' \lts{\comm\pp\val{\level}\lambda{\topic}\q}
\pa\pp{\PP'}\pc\pa\pq{\Q'}\pc \N'$ we get that $\PP \equiv
\procout{\q}{\lambda(\e)}\PP'$ for some $\e$ such that
$\eval{\e}{\val^{\level,\topic}}$, and $\Q\equiv\procin \pp
{\lambda(\x)} \Q_1+\Q_2$.  By
Lemma~\ref{lem:generation}(\ref{lem:generation5}), there are types
$T_P$ and $T_Q$ such that $\derS{}{\PP}{T_\PP}{}$ and
$\derS{}{\Q}{T_\Q}{}$. 
By Lemma~\ref{lem:generation}(\ref{lem:generation2}), $T_P$ must be of
the form $T_P = \tout\q\lambda{\SOT^{\level,\topic}}.\T'_\PP$.  Then
the safety of $\tout\q\lambda{\SOT^{\level,\topic}}.\T'_\PP$ (more
specifically, the premise of Rule \rulename{safe-out}) implies that
$\level\levleq\rho(\participq,\topic)$. This concludes the proof of
Property 1 of session safety (AC).

Suppose now that the above computation continues as follows:
\begin{myformula}{\pa\pp{\PP'}\pc\pa\pq{\Q'}\pc \N' 
 \lts{~\quad\sigma'\quad~}  \pa\pq {\Q_2} \pc \N''
 \lts{\comm\q \valu{\level'}{\lambda'}{\topicpsi}\participr}  \pa\pq \Q''' \pc \N'''
}\end{myformula}
namely, that the trace $\sigma \cdot\comm\pp
\val{\level}\lambda{\topic}\q$ is extended to the relay trace $\sigma \cdot\comm\pp \val{\level}\lambda{\topic}\q \cdot\sigma' \cdot\comm\q \valu{\level'}{\lambda'}{\topicpsi}\participr$.
From
$\Q\equiv\procin \pp  {\lambda(\x)} \Q_1+\Q_2$, by 
Lemma~\ref{lem:generation}(\ref{lem:generation3}) we get $T_\Q = \T_1 \wedge
\T_2$ and $\derS{}{\procin \pp  {\lambda(\x)} \Q_1}{\T_1}{}$. By
applying now 
Lemma~\ref{lem:generation}(\ref{lem:generation1}),
we obtain $\T_1 = \pp ? \lambda(\SOT^{\level,\topic}).\T'_1$ and
$\derS{\x: \SOT^{\level,\topic}}{\Q_1}{\T'_1}{}$. From this, since
$\Q' = \Q_1 \sub{\val^{\level,\topic}}{\x} $, we infer 
$\derS{}{\Q'}{\T'_1}{}$. By Theorem~\ref{thm:SR},
$\derS{}{\Q'}{\T'_1}{}$ implies $\derS{}{\Q''}{\T''_1}{}$ for some
$\T''_1$ such that $\T'_1\RedT\T''_1$.  Now, since \mbox{$\pa\q {\Q''} \pc \N''
  \lts{\comm\q \valu{\level'}{\lambda'}{\topicpsi}\participr} \pa\q
  \Q''' \pc \N'''$}, we have $\Q''\equiv \procout \participr
{\lambda'(\e')} \Q'''$ for some $\e'$ such that
$\eval{\e'}{\valu^{\level',\topicpsi}}$.
By 
Lemma~\ref{lem:generation}(\ref{lem:generation2}), $T''_1 =
{\tout\participr {\lambda'}{\SOT_1^{\level',\topicpsi}}.\T'''_1}{}$.
Now, the
safety of $T_1 = \tin\p \lambda{\SOT^{\level,\topic}}.\T'_1$ 
(and more specifically, the premise of Rule \rulename{safe-in})
implies that $\pair\level\topic\prec T'_1$ and therefore also
$\pair\level\topic\prec T''_1 = \tout\participr
{\lambda'}{\SOT_1^{\level',\topicpsi}}.\T'''_1$, since $\T''_1$ is
obtained by reducing $\T'_1$ (and therefore $\T''_1$ is a subterm of $\T'_1$).
Then $\level \levleq \level'$ or
$\diff\topic\topicpsi$ by definition of agreement (Rule
\rulename{agr-out}).
This concludes the proof of
  Property 2 of session safety (LF).
\end{proof}

\mysection{Related and Future Work}\label{sec:futrel} 
We introduced the notion of topic as a way to relax
  security type systems for session calculi. We focussed on 
  multiparty rather than binary sessions, as security issues
  appear to be less relevant for binary sessions. Indeed, binary
  sessions may often be viewed as client-server interactions,
  where one can assume that the client chooses the server (and thus to
  some extent trusts it) and that the server is protected against
  malicious clients.  On the other hand, in a multiparty session the
  parties are symmetric peers which may not know each other and thus
  require to be protected against each other.
  
  \smallskip

The first multiparty session calculus with synchronous communication
was presented in~\cite{BY09}. Here we considered an enrichment of the
calculus of~\cite{DGJPY15} with security and types. The base calculus
is admittedly very simple, as it cannot describe parallel and
interleaved sessions, and its type system only allows internal choices
among outputs and external choices among inputs. Our version is even
simpler than that of ~\cite{DGJPY15} since 
the syntax does not include the conditional construct. The advantage
of this minimal setting is that the safety property, which covers both
access control and leak freedom, enjoys a particularly simple
definition. In particular, leak freedom amounts to a condition on
\emph{mediators}, which are participants acting as a bridge between a
sender and a receiver. This condition says that after receiving high
information by the sender on some topic, the mediator should not send
low information to the receiver on a related topic.  
%

\smallskip

It can be argued that topics are orthogonal to structured
  communication features, and could therefore be studied in a more
  general setting. However, within a structured communication
  the set of topics is delimited a priori, as specified by the global
  type, so the notion becomes more effective.
  
  \smallskip

  One further issue is that of expressiveness of topics. One may
    wonder whether the use of topics could be simulated by using other
    ingredients of our calculus, such as security levels, labels and
    base types.  Clearly, the independence of topics in the two end
    messages of a relay trace cannot be represented by the
    incomparability of their security levels: a safe relay trace
    $m_1\cdot \sigma\cdot m_2$ where $m_1$ and $m_2$ have unrelated
    topics could not be mimicked by the same trace with incomparable
    security levels for $m_1$ and $m_2$, since the latter is insecure
    in a classical LF approach. As for labels, they are meant to
    represent different options in the choice operators, so they are
    conceptually quite different from topics.

\smallskip

\noindent {\bf Related work.}
Compared to previous work on security-enriched multiparty session
calculi~\cite{CCD14,CCD16}, our definition of leak freedom is more
permissive in two respects: 
\begin{myenumerate}
\item A sequence of messages directed to the same participant is
  always allowed. In the calculi of~\cite{CCD14,CCD16}, where deadlocks could
  arise, it was necessary to prevent any low communication after a
  high communication (because the mere fact that the high
  communication could fail to occur would cause a leak). For instance
  the trace (omitting labels) $\commsimp{\pp}{\q}{v}{\top}{\topic}
  \cdot \commsimp{\pp'}{\q}{u}{\bot}{\topic}$ was rejected in those
  calculi, while it is allowed in the present one, which is
  deadlock-free. In our case it is only the content of a message that
  can be leaked, and therefore it is enough to focus on relay
  sequences made of a message \emph{to} a participant, followed by a
  message \emph{from} the same participant.
\item Thanks to the introduction of topics, the standard leak-freedom
  requirement can be relaxed also on relay sequences, by forbidding
  only downward flows between messages on correlated
  topics.
\end{myenumerate}

\noindent
One could see the use of topics as a way of implementing
  declassification (see \cite{SabSands09} for a survey). For instance,
  a relay trace whose end messages carry values $\val_1^{\top, \topic_1}$ and $\val_2^{\bot, \topic_2}$
  with independent topics $\topic_1$ and $\topic_2$ could be
  interpreted as the application of a {\em trusted function} (such as
  encryption \cite{SabSands09}) to transform a secret value $\val_1$
  into a public value $\val_2$.

\smallskip

\noindent {\bf Future work.}
We intend to explore further the relationship between topics and declassification.
Also, inspired by~\cite{LC15}, we plan to enrich the present calculus by
allowing levels and topics to depend on exchanged values. Indeed,
  it seems reasonable to expect
that a server should 
conform the levels and topics of 
its messages 
to its different kinds of clients.  For
example an ATM should receive credit card numbers with personalised
topics.

\myparagraph{Acknowledgments} 
We are grateful to the anonymous reviewers for their useful remarks.

\bibliographystyle{eptcs}
\bibliography{session-short}

\end{document}